\documentclass[12pt,a4paper]{article}

\usepackage{amssymb}
\usepackage{amsmath}
\usepackage{amsthm}
\usepackage{textcomp}
\usepackage{array}
\usepackage[latin5]{inputenc}
\usepackage{lineno}
\usepackage{color}
\usepackage{verbatim} 
\usepackage[round,semicolon,authoryear]{natbib}
\usepackage{setspace}
\usepackage{endnotes}

\newcommand{\bra}[1]{\langle #1|}
\newcommand{\ket}[1]{|#1\rangle}

\newcommand{\cent}[0]{\mbox{\textcent}}
\newcommand{\dollar}[0]{\$}

\newtheorem{lemma}{Lemma}
\newtheorem{theorem}{Theorem}
\newtheorem{corollary}{Corollary}
\newtheorem{openproblem}{Open Problem}
\newtheorem{fact}{Fact}

\newcommand{\setD}{<\mspace{-4mu}>}
\newcommand{\rev}[0]{\mathrm{r}}
\newcommand{\LRofC}[1]{\left\langle #1 \right\rangle }

\title{Quantum computation with devices whose contents are never read\thanks{
A preliminary version of this work was presented in the 9$ ^{th} $ International Conference on 
Unconventional Computation (UC2010).}}

\author{Abuzer Yakary{\i}lmaz\thanks{Corresponding author}\\
\small Bo\u{g}azi\c{c}i University, Department of Computer Engineering,\\ [-0.8ex]
\small Bebek 34342 \.{I}stanbul, Turkey \\ [-0.8ex]
\small \texttt{abuzer@boun.edu.tr}
\and
R\= usi\c n\v s Freivalds\\
\small Institute of Mathematics and Computer Science, University of Latvia,\\[-0.8ex]
\small Rai\c na bulv\= aris 29, Riga, LV-1459, Latvia\\[-0.8ex]
\small \texttt{Rusins.Freivalds@mii.lu.lv}
\and
A. C. Cem Say\\
\small Bo\u{g}azi\c{c}i University, Department of Computer Engineering,\\ [-0.8ex]
\small Bebek 34342 \.{I}stanbul, Turkey \\ [-0.8ex]
\small \texttt{say@boun.edu.tr}
\and
Ruben Agadzanyan\\
\small Institute of Mathematics and Computer Science, University of Latvia,\\[-0.8ex]
\small Rai\c na bulv\= aris 29, Riga, LV-1459, Latvia\\[-0.8ex]
\small \texttt{agrubi@gmail.com}
}

\date{\small Keywords: quantum computation, write-only memory, quantum finite automata, counter automata, quantum Fourier transform}

\begin{document}

\maketitle

\pagenumbering{arabic}

\thispagestyle{empty}

\begin{abstract}
In classical computation, a ``write-only memory" (WOM) is little more
than an oxymoron, and the addition of WOM to a (deterministic or
probabilistic) classical computer brings no advantage. We prove that quantum computers that are augmented with WOM can solve problems that
neither a classical computer with WOM nor a quantum computer without
WOM can solve, when all other resource bounds are equal. We focus on realtime quantum finite automata, and examine the increase in their power effected by the addition of WOMs with different access modes and capacities. Some problems that are unsolvable by two-way probabilistic Turing machines using sublogarithmic amounts of read/write memory are shown to be solvable by these enhanced automata.
\end{abstract}

\section*{Introduction}

It is well known that many physical processes that violate human
``common sense" are in fact sanctioned by quantum theory.
Quantum computation as a field is interesting for precisely the fact that it
demonstrates that quantum computers can perform resource-bounded tasks
which are (in some cases, provably) beyond the capabilities of
classical computers. In this paper, we demonstrate that the
usage of ``write-only memory" (WOM), a computational component that is
used exclusively for being written to, and never being read, (which is
little more than a joke in the classical setup,) can improve the power of
quantum computers significantly. We prove that a
quantum computer
using a WOM can solve problems that neither a classical computer with
a WOM nor a quantum computer without a WOM can solve, when all
machines are restricted to use a constant amount of read/write memory;
in fact, we show that certain tasks that cannot be achieved by
classical probabilistic Turing machines, even if they are allowed to
use sublogarithmic amounts of read/write memory, can be performed by
quantum finite automata augmented with write-only memory (QFA-WOMs).

In the rest of the paper, we first present the basic notation and
terminology that will be employed. After a review of the automaton
variants whose powers will be compared with those of the machines with
write-only memory, we give a formal definition of the automata
with WOM. Following a thorough examination of the power of a QFA-WOM
variant with severe restrictions on the way in which it can access its
WOM, we look at less restricted models. The interesting question
of just how much WOM is sufficient to endow an otherwise finite-state
machine with the capability of recognizing nonregular languages is addressed, and a way of looking at our results as contributions
about quantum function computation is presented.

\section*{Preliminaries} \label{preliminaries}

For a given vector $ v $, $ v[i] $ is the $ i^{th} $ entry of $ v $ and for a given string $ w $,
$ |w| $ is the length of $ w $.

$ \Sigma $ (\textit{input alphabet}): $ \Sigma $ is a finite set of symbols,
	i.e. $ \Sigma = \{ \sigma_{1}, \ldots, \sigma_{|\Sigma|} \} $.
	As a convention, $ \Sigma $ never contains the symbol $ \boldmath{\#} $ (\textit{the blank symbol}),
	$ \mathbf{\cent} $ (\textit{the left input end-marker}), and 
	$ \dollar $ (\textit{the right input end-marker}).
	$ \tilde{\Sigma} $ denotes the set $ \Sigma \cup \{\cent,\dollar\} $.
	Additionally, $ \tilde{w} $ denotes the string $ \cent w \dollar $, 
	for any given input string $ w \in \Sigma^{*} $.

$ Q $ (\textit{the set of internal states}): $ Q $ is a finite set of internal states,
	i.e. $ Q = \{ q_{1},\ldots,q_{|Q|} \} $. 
	Unless otherwise specified, $ q_{1} $ is \textit{the initial state}.
	Moreover, depending on the context, $ Q_{a} \subseteq Q $ (resp., $ Q \setminus Q_{a} $) 
	is the set of accepting (resp., rejecting) states.
	
$ \setD $ (\textit{the set of head directions}): 
	$ <\mspace{-4mu}> $ is the set $ \{ \leftarrow, \downarrow , \rightarrow \} $, where 
	``$ \leftarrow $" means that the (corresponding) head moves one square to the left,
	``$ \downarrow $" means that the head stays on the same square, and
	``$ \rightarrow $" means that the head moves one square to the right.
	As a special case, $ \rhd $ is the set $ \{ \downarrow,\rightarrow \} $.

$ \Theta $ (\textit{the counter status}): $ \Theta $ is the set $ \{ 0,1 \} $, where $ 1 $ means that the counter 	
	value is nonzero, and $ 0 $ means that the counter value is zero.
	
$ \delta $ (\textit{the transition function}): The behavior of a machine is specified by
	its transition function. The domain and the range of a transition function may vary with respect to
	the capabilities of the model.

$ f_{\mathcal{M}}^{a}(w) $ (\textit{acceptance probability}): For a given machine $ \mathcal{M} $
	and an input string $ w \in \Sigma^{*} $, $ f_{\mathcal{M}}^{a}(w) $, or shortly $ f_{\mathcal{M}}(w) $, 
	is the probability that $ w $ will be accepted by $ \mathcal{M} $.
	Moreover, $ f_{\mathcal{M}}^{r}(w) $ will be used 
	in order to represent \textit{the rejection probability} of $ w $ by $ \mathcal{M} $.

The language $ L \subset \Sigma^{*} $ recognized by machine $
\mathcal{M} $ with (strict) cutpoint
$ \lambda \in \mathbb{R} $ is defined as
\begin{equation}
       L = \{ w \in \Sigma^{*} \mid f_{\mathcal{M}}(w) > \lambda \}.
\end{equation}
The language $ L \subset \Sigma^{*} $ recognized by machine $
\mathcal{M} $ with nonstrict cutpoint
$ \lambda \in \mathbb{R} $ is defined as \citep{BJKP05}
\begin{equation}
       L = \{ w \in \Sigma^{*} \mid f_{\mathcal{M}}(w) \geq \lambda \}.
\end{equation}
The language $ L \subset \Sigma^{*} $ is said to be recognized by
machine $ \mathcal{M} $ with unbounded error
if there exists a cutpoint  $ \lambda \in \mathbb{R} $ such that
$ L $ is recognized by $ \mathcal{M} $ with strict or nonstrict
cutpoint $ \lambda $.

The language $ L \subset \Sigma^{*} $ recognized by machine $
\mathcal{M} $ with error bound $ \epsilon $
($ 0 \le \epsilon < \frac{1}{2} $) is defined as
\begin{itemize}
       \item $ f_{\mathcal{M}}^{a}(w) \ge 1 - \epsilon $ when $ w \in L $,
       \item $ f_{\mathcal{M}}^{r}(w) \ge 1 - \epsilon $ when $ w \notin L $.
\end{itemize}
This situation is also known as recognition with bounded error.

The language $ L \subset \Sigma^{*} $ is said to be recognized by machine $ \mathcal{M} $ 
with (positive) one-sided bounded error if there exists a $ p \in (0,1] $ such that
\begin{itemize}
	\item $ f_{\mathcal{M}}(w) \geq p $ when $ w \in L $ and
	\item $ f_{\mathcal{M}}(w) = 0 $ when $ w \notin L $.
\end{itemize}
Equivalently, it can be said that $ L \subset \Sigma^{*} $ is recognized by machine $ \mathcal{M} $ 
with (positive) one-sided error bound $ \epsilon $, where $ \epsilon = 1-p $ (and so $ \epsilon \in [0,1) $).

The language $ L \subset \Sigma^{*} $ is said to be recognized by machine $ \mathcal{M} $ 
with negative one-sided bounded error if there exists a $ p \in (0,1] $ such that
\begin{itemize}
	\item $ f_{\mathcal{M}}(w) = 1 $ when $ w \in L $ and
	\item $ f_{\mathcal{M}}(w) \le 1-p $ when $ w \notin L $.
\end{itemize}
Equivalently, it can be said that $ L \subset \Sigma^{*} $ is recognized by machine $ \mathcal{M} $ 
with negative one-sided error bound $ \epsilon $, where $ \epsilon = 1-p $ (and so $ \epsilon \in [0,1) $).
	
\section*{Conventional computation models}

In this section, we will review the conventional computational models (i.e. those not involving write-only memory) to be used in the paper. The reader is assumed to be familiar with the  standard definitions (involving a read-only input tape and one
read/write work tape) \citep{AB09} for
deterministic and probabilistic Turing machines (TM and PTM, respectively). 
Our definition of quantum Turing machine (QTM) is a
modification of the one found in \citep{Wa98}\endnote{The QTM model is appropriate for studying the
effect of space bounds on computational power. See \citep{Ya93} for an
alternative model of quantum computation.}.
Technically, we allow the additional finite register, which is observed after
each step of the computation to decide whether to accept, reject, or
continue, to have multiple symbols in its alphabet corresponding to
each of  these alternatives, and to be refreshed to its initial symbol
after each observation\endnote{Unlike \citep{Wa98}, 
we also allow efficiently computable
irrational numbers as transition amplitudes in our QTM's. This
simplifies the description of some algorithms in the remainder of this
paper.}. The addition of this finite register, which will be explained in greater detail in the context of the definition of quantum counter automata below, allows our QTM's to
implement general quantum operations, and therefore to simulate their
classical counterparts precisely and efficiently. This result was
shown for QTM's with classical tape head position by Watrous \citep{Wa03}.

Our discussion will focus almost entirely on realtime computation, where the input is consumed one symbol per step in a single left-to-right pass, and the decision is announced immediately upon the reading of the right end-marker.
A realtime $ k \in \mathbb{Z}^{+} $ counter automaton  (RT-$ k $CA) is a 
realtime finite state automaton augmented with $ k $ counters which can be modified by some amount
from $ \lozenge = \{ -1, 0 , +1 \} $ (``$ -1 $" means that the value of the counter is decreased by 1, ``$ 0 $" means that the value 
	of the counter is not changed, and ``$ +1 $" means that the value 
	of the counter is increased by 1), and where the signs of these counters are also taken into account during transitions.

For a given input string  $ w \in \Sigma $ ($ \tilde{w} $ is on the input tape),
a configuration of a RT-$ k $CA is composed of the following elements:
\begin{itemize}
	\item the current internal state,
	\item the position of the input head, and
	\item the contents of the counters.
\end{itemize}
The computation begins with the initial configuration,
in which the internal state is $ q_{1} $, the value(s) of the counter(s) is (are) zero(s), and
the input head is placed on symbol $ \cent $.

Formally, a realtime probabilistic $ k $-counter automaton
(RT-P$ k $CA) $ \mathcal{M} $ is a 5-tuple\endnote{The reader may find it useful to consult the descriptions of
the common notational items in the discussion, given immediately after
the introduction.}
\begin{equation}
	\mathcal{P} = (Q,\Sigma,\delta,q_{1},Q_{a}).
\end{equation}
The transitions of $ \mathcal{M} $ are specified by $\delta$ as follows: $ \delta(q,\sigma,\bar{\theta},q') $ is the probability that $ \mathcal{M} $ will change its state to $ q' \in Q $,  if it is originally  in state $ q \in Q $, scanning the symbol $ \sigma \in \tilde{\Sigma} $ 
	 on the input tape, and sensing $ \bar{\theta} \in \Theta^{k} $ in its counter(s) (i.e.
	$ \bar{\theta}[i] $ is the status of the $ i^{th} $ counter, where $ 1 \le i \le k $).
In each transition, the input tape head moves one square to the right, and the counters are updated with respect to $ \bar{c} = D_{c}(q') $, where   $ D_{c} $ is a function from $ Q $ to $ \lozenge^{k} $, (i.e.
	the value of the $ j^{th} $ counter is updated by $ \bar{c}[j] $, where $ 1 \le j \le k $).

The input string is accepted by a RT-P$ k $CA 
if the computation 
ends in an accepting state.

A realtime deterministic $ k $-counter automaton
(RT-D$ k $CA) is just a RT-P$ k $CA in which all transitions with nonzero probability have probability 1.

A realtime nondeterministic $ k $-counter automaton
(RT-N$ k $CA) is just a RT-P$ k $CA that is interpreted to recognize a language with cutpoint 0, that is, the language it recognizes is the set of all and only the strings that it accepts with nonzero probability.

For the quantum case, we will only be concerned with one-counter automata.
A realtime quantum $ 1 $-counter automaton (RT-Q$ 1 $CA)\endnote{Note that our definition of quantum counter automata
is more general than the previous ones, \citep{Kr99,BFK01,YKTI02,YKI05} since it is based on
general quantum operators.
} is a 6-tuple
\begin{equation}
	\mathcal{M} = (Q,\Sigma,\Omega,\delta,q_{1},Q_{a}).
\end{equation}
	 When $ \mathcal{M} $ is in state $ q \in Q $, reading symbols $ \sigma \in \tilde{\Sigma} $ 
	 on the input tape,
	and sensing $ \theta \in \Theta $ on the counter, it changes its state to $ q' \in Q $, 
	updates the value of the counter with respect to $ c \in \lozenge  $,
	moves the input tape head one square to the right, and
	writes $ \omega \in \Omega $ in the finite register with transition amplitude
	$ \delta(q,\sigma,\theta,q',c,\omega) \in \mathbb{C} $, satisfying the well-formedness condition to be described below.

As seen above, the quantum machines in this paper will be distinguished from their classical counterparts
by the presence of the item $\Omega$ (\textit{the finite 
register alphabet}) in their definitions. This register, whose incorporation in a classical machine would not make any change to its computational power, is an essential part of our quantum models. In realtime computation, the usage of the finite register can be simplified so that the intermediate observations, mentioned above for the case of general QTMs, are not required, and a single measurement of the internal state at the end suffices for our purposes \citep{Ya10A,YS10C}. 
$ Q $ is partitioned into two disjoint sets, $ Q_{a} $, and $ Q_{r} = Q \setminus Q_{a} $.
In each transition, the quantum machine goes through the following phases:
\begin{enumerate}
	\item \textit{pre-transition phase}: 
		reset the register to its initial symbol ``$ \omega_{1} $";
	\item \textit{transition phase}:
		update the content of the register, in addition to the changes in the configuration components normally associated by the transitions of the corresponding classical machine.
\end{enumerate}

On the right end-marker, the projective measurement 
\begin{equation}
	P = \{ P_{\tau \in \{a, r\}} \mid P_{\tau} = \sum_{q \in Q_{\tau}} \ket{q}\bra{q} \}
\end{equation}
is performed. The result associated with $P_{\tau} $ is simply $\tau$, and the input is accepted if ``$a$" is observed. Note that this just means that the acceptance probability is the sum of the squares of the moduli of the amplitudes of the accepting states at the end of the computation.

Since we do not consider the register content as part of the
configuration, the register can be seen as the
``environment" interacting with the ``principal system" that is the rest of the quantum machine \citep{NC00}.
$ \delta $ therefore induces a set of configuration transition
matrices, $ \{ E_{\omega \in \Omega} \} $, where the $ (i,j)^{th} $ entry of $ E_{\omega}$,
the amplitude of the transition  from $ c_{j} $ to $ c_{i} $ by
writing $ \omega \in \Omega $ on the register,
is defined by $ \delta $ whenever $ c_{j} $ is reachable from $ c_{i}
$  in one step, and is zero otherwise. 
The $ \{ E_{\omega \in \Omega} \} $ form an operator $ \mathcal{E} $.

Let $ \mathcal{C} $ be the set of configurations that can be attained by the machine for a given input string.
According to the modern understanding of quantum computation \citep{AKN98}, 
a quantum machine is said to be \textit{well-formed}
if $ \mathcal{E} $ is a superoperator (selective quantum operator), i.e.
\begin{equation}
      \sum_{\omega \in \Omega} E_{\omega}^{\dagger}E_{\omega} = I.
\end{equation}
$ \mathcal{E} $ can be represented by a $ | \mathcal{C} | |
\Omega | \times | \mathcal{C} | $-dimensional
matrix $ \mathsf{E} $ (Figure \ref{figure:matrix-E}) by concatenating each $ E_{\omega \in \Omega} $ one under the
other. It can be verified that $ \mathcal{E} $ is a superoperator if and only
if the columns of $ \mathsf{E} $ form an orthonormal set.

\begin{center}
\begin{figure}[h!]
	\centering
	\caption{Matrix $ \mathsf{E} $}
	\begin{minipage}{0.4\textwidth}
		\begin{equation}
		\begin{array}{ccccc}
			\multicolumn{1}{c|}{} & c_{1} & c_{2} & \ldots & \multicolumn{1}{c|}{ c_{|\mathcal{C}|} } \\
			\hline
			\multicolumn{1}{c|}{c_{1}} & & & & \multicolumn{1}{c|}{} \\
			\multicolumn{1}{c|}{c_{2}} & & & & \multicolumn{1}{c|}{} \\
			\multicolumn{1}{c|}{\vdots} & \multicolumn{4}{c|}{ E_{\omega_{1}} } \\
			\multicolumn{1}{c|}{c_{|\mathcal{C}|}} & & & & \multicolumn{1}{c|}{} \\
			\hline	
			\multicolumn{1}{c|}{c_{1}} & & & & \multicolumn{1}{c|}{} \\
			\multicolumn{1}{c|}{c_{2}} & & & & \multicolumn{1}{c|}{} \\
			\multicolumn{1}{c|}{\vdots} & \multicolumn{4}{c|}{ E_{\omega_{2}} } \\
			\multicolumn{1}{c|}{c_{|\mathcal{C}|}} & & & & \multicolumn{1}{c|}{} \\
			\hline\multicolumn{1}{c|}{c_{1}} & & & & \multicolumn{1}{c|}{} \\
			\multicolumn{1}{c|}{c_{2}} & & & & \multicolumn{1}{c|}{} \\
			\multicolumn{1}{c|}{\vdots} & \multicolumn{4}{c|}{\vdots } \\
			\multicolumn{1}{c|}{c_{|\mathcal{C}|}} & & & & \multicolumn{1}{c|}{} \\
			\hline	
			\multicolumn{1}{c|}{c_{1}} & & & & \multicolumn{1}{c|}{} \\
			\multicolumn{1}{c|}{c_{2}} & & & & \multicolumn{1}{c|}{} \\
			\multicolumn{1}{c|}{\vdots} & \multicolumn{4}{c|}{ E_{\omega_{|\Omega|}} } \\
			\multicolumn{1}{c|}{c_{|\mathcal{C}|}} & & & & \multicolumn{1}{c|}{} \\
			\hline	
		\end{array}
		\end{equation}
	\end{minipage}
	\label{figure:matrix-E}
\end{figure}
\end{center}

We define a realtime quantum finite automaton (RT-QFA) as just a RT-Q1CA that never updates its counter. 
The class of languages recognized with bounded error by RT-QFAs equals the class of regular languages
\citep{Bo03,Je07,AY10A}.

\begin{lemma}
	\label{lem:classical-simulated-by-quantum}
	Any classical automaton can be simulated by a quantum automaton of the corresponding type exactly,
	such that the simulating and simulated machines agree on the value of the acceptance probability of any string.
\end{lemma}
\begin{proof}
	See \citep{Pa00,Hi08,Wa09,AY10A,Ya10A,YS10C}.
 \end{proof}

The next lemma demonstrates a useful programming trick about counters.
(We will show this for the quantum case. It is well known that the
same method can also be employed for classical counter machines.) For
any counter automaton model A, let A($m$) be a machine of type A with
the additional ability of updating each of its counters with an
increment from the set $ \{ -m,\ldots,m \} $, where $ m>1 $, in a single step.

\begin{lemma}
       \label{wom:lem:inc-m-equal-inc-1}
       For any RT-Q1CA($ m $) $ \mathcal{M} $, there exists a corresponding
RT-Q1CA $ \mathcal{M'} $ such that
       \begin{equation}
               f_{\mathcal{M}}(w)=f_{\mathcal{M'}}(w),
       \end{equation}
       for all $ w \in \Sigma^{*} $, where $ m > 1 $.
\end{lemma}
\begin{proof}
       Let $ \mathcal{M} = (Q,\Sigma,\Omega,\delta,q_{1},Q_{a}) $. We construct
 $ \mathcal{M}' = (Q',\Sigma,\Omega,\delta',q_{1}',Q_{a}') $.
$ Q' $ contains $ m $ states    for each  state of $ \mathcal{M} $, that
is, the states of $ \mathcal{M}' $ are of the form
       \begin{equation}
               \LRofC{q,i} \in Q \times \{0,\ldots,m-1\}.
       \end{equation}
        Moreover, $ q_{1}' = \LRofC{q_{1},0} $, and
       \begin{equation}
               Q_{a}' = \{ \LRofC{q,i} \mid q \in Q_{a}, i \in \{0,\ldots,m-1\} \}.
       \end{equation}
       Let
       \begin{equation}
               \varphi : \mathbb{Z} \rightarrow \mathbb{Z} \times \{ 0,\ldots,m-1\}
       \end{equation}
       be a bijection such that
       \begin{equation}
               \varphi(x)=
               \left( \left\lfloor \frac{x}{m} \right\rfloor, (x \mod m) \right).
       \end{equation}
       Hence, we can say that the counter values of $ \mathcal{M} $, say $ x
\in \mathbb{Z} $,
       can be equivalently represented by $ \varphi(x) $, based on which we
will construct $ \mathcal{M}' $,
       where $ \varphi(x)[1] $ is stored by the counter, and $ \varphi(x)[2] $
       is stored by the internal state.
       That is, for any configuration of $ \mathcal{M} $, say ($ q,x $),
       \begin{equation}
               ( \LRofC{ q , \varphi(x)[2]},
                       \varphi(x)[1])
       \end{equation}
is an equivalent configuration of $ \mathcal{M}' $.
       Moreover, the transitions of $ \mathcal{M}' $ can be obtained from
those of $ \mathcal{M} $
       in the following way:
       for any $ i \in \{-m,\ldots,m\} $
       and $ j \in \{0,\ldots,m-1\} $,
       the part of the transition
       \begin{equation}
               (q,\sigma) \overset{\delta}{\longrightarrow} \alpha (q',i,\omega)
       \end{equation}
       of $ \mathcal{M} $ is replaced by transition
       \begin{equation}
               ( \LRofC{q,j},\sigma) \overset{\delta'}{\longrightarrow}
               \alpha \left( \LRofC{ q', j+i \mod m)},
               \left\lfloor \frac{j+i}{m} \right\rfloor,\omega \right)
       \end{equation}
       in $ \mathcal{M}' $, where $ q \in Q $, $ \sigma \in \tilde{\Sigma}
$, $\omega \in \Omega $, and
       $ \alpha \in \mathbb{C} $ is the amplitude of the transition.
       Since $ \varphi $ is a bijection, the configuration matrix of $
\mathcal{M} $ is isomorphic to
       that of $ \mathcal{M}' $ for any input string $ w \in \Sigma^{*} $.
       Therefore, they carry out exactly the same computation on a given
input string, say $ w \in \Sigma^{*} $
       and so
       \begin{equation}
               f_{\mathcal{M}}(w) = f_{\mathcal{M}'}(w).
       \end{equation}
 \end{proof}

An $ r $-reversal RT-$ k $CA \citep{Ch81}, denoted as $ r $-rev-RT-$ k $CA, is a RT-$ k $CA where
the number of alternations from increasing to decreasing 
and vice versa on each counter is restricted by $ r $, where $ r $ is a nonnegative integer.

A RT-$ k $CA with blind counters, denoted a RT-$ k $BCA, is a RT-$ k $CA that never checks the status of its counter(s),
(and so the component $ \Theta^{k} $ is completely removed from $ \delta $,) and accepts its input only if all counters are zero, and the processing of the input has ended in an accept state.

\begin{fact}
	\citep{Fr79}
	For every $k$, if $ L $ is recognized by a deterministic RT-$ k $BCA (RT-D$ k $BCA), then for every $ \epsilon \in (0, \frac{1}{2}) $, then
	there exists a probabilistic RT-1BCA (RT-P1BCA) recognizing $L$ with negative one-sided error bound $ \epsilon $.
\end{fact}
We can generalize this result to probabilistic RT-$ k $BCAs (RT-P$ k $BCAs), where $ k>1 $.
\begin{lemma}
	\label{wom:lem:RT-PkBCA-by-RT-P1BCA}
	Let $ \mathcal{P} $ be a given RT-P$ k $BCA and $ \epsilon \in (0,\frac{1}{2}) $ be a given error bound. 
	Then, there exists a RT-P1BCA($ R $) $ \mathcal{P}' $ such that for all $ w \in \Sigma^{*} $,
	\begin{equation}
		f_{\mathcal{P}}(w) \le f_{\mathcal{P}'}(w) \le f_{\mathcal{P}}(w)+\epsilon(1-f_{\mathcal{P}}(w)),
	\end{equation}
	where $ R = 2^{\left\lceil \frac{k}{\epsilon} \right\rceil} $.
\end{lemma}
\begin{proof} 
	Based on the probabilistic method described in Figure \ref{wom:fig:fr-79},
	we can obtain $ \mathcal{P}' $ by making the following modifications on $ \mathcal{P} $:
	\begin{figure}[h!]
       \caption{Probabilistic zero-checking of multiple counters by one counter}
       \centering
       \fbox{
       \begin{minipage}{0.9\textwidth}
               \small
               In this figure, we review a method presented by Freivalds in \citep{Fr79}:
               Given a machine with $ k>1 $ counters, say $ C_{1},\ldots,C_{k}  $,
whose values
               can be updated using the increment set $ \{-1,0,1\} $, we can build
a machine with a single counter, say $ C $,  whose value can be
updated using the increment set $ \{ -R,\ldots,R \} $
               ($ R = 2^{\left\lceil \frac{k}{\epsilon} \right\rceil} $), such that
               all updates on $ C_{1},\ldots,C_{k} $ can be simulated on $ C $ in
the sense that
               (i) if all values of $ C_{1},\ldots,C_{k} $ are zeros, then the
value of $ C $ is zero; and
               (ii) if the value of at least one of $ C_{1},\ldots,C_{k} $ is nonzero, then
               the value of $ C $ is nonzero with probability $ 1-\epsilon $,
               where $ \epsilon \in (0,\frac{1}{2}) $.
               The  probabilistic method for this simulation is as follows:
               \begin{itemize}
                       \item Choose a number $ r $ equiprobably from the set $ \{1,\ldots,R\} $.
                       \item The value of $ C $ is increased (resp., decreased) by $ r^{i} $
                               if the value of $ C_{i} $ is increased (resp., decreased) by 1.
               \end{itemize}
       \end{minipage}
       }
       \label{wom:fig:fr-79}
\end{figure}
	\begin{enumerate}
		\item At the beginning of the computation, $ \mathcal{P}' $ equiprobably 
			chooses a number $ r $ from the set $ \{1,\ldots,R\} $.
		\item For each transition of $ \mathcal{P} $, in which the values of counters are updated
		by $ (c_{1},\ldots,c_{k}) \in \{-1,0,1\}^{k} $, i.e., the value of the $ i^{th} $ counter
		is updated by $ c_{i} $ ($ 1 \le i \le k $), $ \mathcal{P} $ 
		makes the same transition by updating its counter values by $ \sum \limits_{i=1}^{k} r^{i}c_{i} $.
	\end{enumerate}
	Hence, (i) for each accepting path of $ \mathcal{P} $, the input is accepted by $ \mathcal{P}' $, too;
	(ii) for each rejecting path of $ \mathcal{P} $, the input may be accepted by $ \mathcal{P}' $
	with a probability at most $ \epsilon $.
	By combining these cases, we obtain the following inequality for any input string $ w \in \Sigma^{*} $: 
	\begin{equation}
		f_{\mathcal{P}}(w) \le f_{\mathcal{P}'}(w) \le f_{\mathcal{P}}(w)+\epsilon(1-f_{\mathcal{P}}(w))
	\end{equation}
 \end{proof}
\begin{theorem}
	\label{wom:thm:RT-PkBCA-by-RT-P1BCA}
	If $ L $ is recognized by a RT-P$ k $BCA
	with error bound $ \epsilon \in (0, \frac{1}{2}) $, then
	$ L $ is recognized by a RT-P1BCA with error bound 
	$ \epsilon' $ ($ 0 < \epsilon < \epsilon' < \frac{1}{2} $). 
	Moreover, $ \epsilon' $ can be tuned to be arbitrarily close to $ \epsilon $.
\end{theorem}
\begin{proof}
	Let $ \mathcal{P} $ be a RT-P$ k $BCA recognizing $ L $ with error bound $ \epsilon $.
	By using the previous lemma (Lemma \ref{wom:lem:RT-PkBCA-by-RT-P1BCA}), 
	for any $ \epsilon'' \in (0,\frac{1}{2}) $,
	we can construct a RT-P1BCA($ R $), say $ \mathcal{P}'' $, from $ \mathcal{P} $,
	where $ R=2^{\left\lceil \frac{k}{\epsilon''} \right\rceil} $.
	Hence, depending on the value of $ \epsilon $,
	we can select $ \epsilon'' $ to be sufficiently small such that
	$ L $ is recognized by $ \mathcal{P}'' $ with error bound 
	$ \epsilon'=\epsilon+\epsilon''(1-\epsilon) < \frac{1}{2} $.
	Since for each RT-P1BCA($ m $), there is an equivalent RT-P1BCA for any $ m>1 $,
	$ L $ is also recognized by	a RT-P1BCA 
	with bounded error $ \epsilon' $, which can be tuned to be arbitrarily close to $ \epsilon $.
 \end{proof}
\begin{corollary}
	If $ L $ is recognized by a RT-P$ k $BCA
	with negative one-sided error bound $\epsilon \in (0, 1) $, then
	$ L $ is recognized by a RT-P1BCA with negative one-sided error bound 
	$ \epsilon' $ ($ 0 < \epsilon  < \epsilon' $). 
	Moreover, $ \epsilon' $ can be tuned to be arbitrarily close to $ \epsilon  $.
\end{corollary}

\section*{Models with write-only memory}

We model a WOM as a two-way write-only tape having alphabet $ \Gamma = \{\gamma_{1},\ldots,\gamma_{|\Gamma|}\} $.
	$ \Gamma $ contains $ \# $ and $ \varepsilon $ (\textit{the empty string} or 
	\textit{the empty symbol}).

In a general two-way WOM, a transition ``writing" $ \varepsilon $ on the tape causes the tape head to move in the specified direction without changing the symbol in the square under the original location of the head.
If we restrict the tape head movement of a WOM to $ \rhd $, i.e. one-way,
we obtain a ``push-only stack" (POS).
In the case of machines with POS, we assume that the write-only tape head does not move if a $ \varepsilon $ is written, and it moves one square to the right
if a symbol different than $ \varepsilon $ is written.
A special case of the POS setup is the ``increment-only counter" (IOC) \citep{SYY10A}, where $ \Gamma $ includes only $ \varepsilon $ and a single \textit{counting} symbol 
(different than $ \# $).

For any standard machine model, say, M, we use the name M-WOM to
denote M augmented with a WOM component. A TM-WOM, for instance, just has an additional
write-only tape.
The computational power of the PTM-WOM is easily seen to be the same
as that of the PTM; since the machine does not use the contents of the
WOM in any way when it decides what to do in the next move, and the probability distribution of the partial machine configurations (not including the WOM) is not dependent on the WOM content at any time, every
write-only action can just as well be replaced with a write-nothing
action. The following lemma shows this more formally for the models that will come under focus in this paper.

\begin{lemma}
	\label{wom:lem:classical-wom}
	The computational power of any realtime classical finite automaton is unchanged
	when the model is augmented with a WOM.
\end{lemma}
\begin{proof}
	For a given machine $ \mathcal{M} $ and an input string $w$, consider the tree $ \mathcal{T} $ of states, where the root is the initial state, each subsequent level corresponds to the processing of the next input symbol, and the children of each node $N$ are the states that have nonzero-probability transitions from $S$ with the input symbol corresponding to that level. Each such edge in the tree is labeled with the corresponding transition probability. The probability of node $N$ is the product of the probabilities on the path to $N$ from the root. The acceptance probability is the sum of the probabilities of the accept states at the last level.

Now consider attaching a WOM to $ \mathcal{M} $, and augmenting its program so that every transition now also specifies the action to be taken on the WOM. Several new transitions of this new machine may correspond to a single transition of $ \mathcal{M} $, since, for example, a transition with probability $p$ can be divided into two transitions with probability $ p \over 2 $, whose effects on the internal state are identical, but which write different symbols on the WOM. It is clear that many different programs can be obtained by augmenting $ \mathcal{M} $ in this manner with different WOM actions. Visualize the configuration tree $ \mathcal{T}_{new} $ of any one of these new machines on input $w$. There exists a homomorphism $h$ from $ \mathcal{T}_{new} $ to $ \mathcal{T} $, where $h$ maps nodes in $ \mathcal{T}_{new} $ to nodes on the same level in $ \mathcal{T} $, the configurations in $h^{-1}(N)$ all have $N$ as their states, and the total probability of the members of $h^{-1}(N)$ equals the probability of $N$ in $ \mathcal{T} $, for any $N$. We conclude that all the machines with WOM accept $w$ with exactly the same probability as $w$, so the WOM does not make any difference.
 \end{proof}

As stated, our aim is to show that WOMs do increase the power of QTMs. We will focus on several variants of realtime quantum finite automata with WOM
(RT-QFA-WOMs), which are just QTM-WOMs which do not use their work
tapes, and move the input tape head to the right in every step.

More precisely, we will examine the power of RT-QFAs that are augmented with WOM, POS, or IOC,
namely, the models RT-QFA-WOM, RT-QFA-POS, or RT-QFA-IOC (0-rev-RT-Q1CA), respectively.
Note that RT-QFA-WOMs, RT-QFA-POSs, and RT-QFA-IOCs are special cases of
quantum realtime Turing machines, pushdown automata, and 
one-counter automata, respectively.

Formally, a RT-QFA-WOM $ \mathcal{M} $ is a 7-tuple
\begin{equation}
	\label{def:RT-QFA-WOM}
	(Q,\Sigma,\Gamma,\Omega,\delta,q_{1},Q_{a}).
\end{equation}
	When in state $ q \in Q $ and reading symbol $ \sigma \in \tilde{\Sigma} $ 
	on the input tape, $ \mathcal{M} $ changes its state to $ q' \in Q $, writes $ \gamma \in \Gamma $ 
	and $ \omega \in \Omega $ on the WOM tape and the finite register, respectively,
	and then updates the position of the WOM tape head with respect to $ d_{w} \in \setD $
	with transition amplitude $ \delta(q,\sigma,q',\gamma,d_{w},\omega) = \alpha $, where
	$ \alpha \in \mathbb{C} $ and $ |\alpha| \leq 1 $.
	
	In order to represent all transitions from the case where $ \mathcal{M} $ is in state $ q \in Q $ and reading
	symbol $ \sigma \in \tilde{\Sigma} $ together, we will use the notation
	\begin{equation}
		\delta(q,\sigma) =  \sum_{(q',\gamma,d_{w},\omega) \in Q \times \Gamma \times \setD \times \Omega }
			\delta(q,\sigma,q',\gamma,d_{w},\omega) (q',\gamma,d_{w},\omega),
	\end{equation}
	where
	\begin{equation}
		\sum_{(q',\gamma,d_{w},\omega) \in Q \times \Gamma \times \setD \times \Omega }
			|\delta(q,\sigma,q',\gamma,d_{w},\omega)|^{2} = 1.
	\end{equation}

A configuration of a RT-QFA-WOM is the collection of
\begin{itemize}
	\item the internal state of the machine,
	\item the position of the input tape head,
	\item the contents of the WOM tape, and the position of the WOM tape head.
\end{itemize}

The formal definition of the RT-QFA-POS is similar to that of the RT-QFA-WOM,
except that the movement of the WOM tape head is restricted to $ \rhd $, and so
the position of that head does not need to be a part of a configuration.
On the other hand, the definition of the RT-QFA-IOC can be simplified by removing the $ \Gamma $
component from (\ref{def:RT-QFA-WOM}):

A RT-QFA-IOC $ \mathcal{M} $ is a 6-tuple
\begin{equation}
	(Q,\Sigma,\Omega,\delta,q_{1},Q_{a}).
\end{equation}
	When in state $ q \in Q $, and reading symbol $ \sigma \in \tilde{\Sigma} $ 
	on the input tape,  $ \mathcal{M} $ changes its state to $ q' \in Q $, writes $ \omega $ in the register, and updates the value of its counter
	by $ c \in \vartriangle = \{0,+1\} $
	with transition amplitude $ \delta(q,\sigma,q',c,\omega) = \alpha $, where
	$ \alpha \in \mathbb{C} $ and $ |\alpha| \leq 1 $.
	
	In order to show all transitions from the case where $ \mathcal{M} $ is in state $ q \in Q $ and reads
	symbol $ \sigma \in \tilde{\Sigma} $ together, we use the notation
	\begin{equation}
		\delta(q,\sigma) =  \sum_{(q',c,\omega) \in Q \times \vartriangle \times \Omega }
			\delta(q,\sigma,q',c,\omega) (q',c,\omega),
	\end{equation}
	where
	\begin{equation}
		\sum_{(q',c,\omega) \in Q \times \vartriangle \times \Omega }
			|\delta(q,\sigma,q',c,\omega)|^{2} = 1.
	\end{equation}

A configuration of a RT-QFA-IOC is the collection of
\begin{itemize}
       \item the internal state of the machine,
       \item the position of the input tape head, and
       \item the value of the counter.
\end{itemize}

\section*{Increment-only counter machines}

We are ready to start our demonstration of the superiority of quantum computers with WOM over those without WOM. We will examine the capabilities of RT-QFA-IOCs in both the bounded and unbounded error settings, and show that they can simulate a family of conventional counter machines, which are themselves superior to RT-QFAs, in both these cases. 

\subsection*{Bounded error} \label{wom:ioc:bounded}

The main theorem to be proven in this subsection is

\begin{theorem}
	\label{wom:thm:RT-Q1BCA-by-RT-QFA-IOC}
	The class of languages recognized with bounded error by RT-QFA-IOCs contains all languages recognized with 
	bounded error by conventional realtime quantum automata with one blind counter (RT-Q1BCAs).
\end{theorem}

Before presenting our proof of Theorem \ref{wom:thm:RT-Q1BCA-by-RT-QFA-IOC}, let us demonstrate the underlying idea by showing how RT-QFA-IOCs can simulate a simpler family of machines, namely, deterministic automata with one blind counter. Define a RT-QFA-IOC($ m $) as a RT-QFA-IOC with the capability of incrementing its counter by any 
value from the set $ \{0,\ldots,m\} $, where $ m>1 $.

\begin{lemma}
	\label{wom:lem:RT-D1BCA-by-RT-QFA-IOC}
	If a language $ L $ is recognized by a RT-D1BCA,
	then $ L $ can also be recognized by a RT-QFA-IOC 
	with negative one-sided error bound $ \frac{1}{m} $, for any desired value of $m$.
\end{lemma}
\begin{proof}
We will build a RT-QFA-IOC($ m $) that recognizes $L$, which is sufficient by Lemma \ref{wom:lem:inc-m-equal-inc-1}. 

	Throughout this proof, the symbol ``$ i $" is reserved for the imaginary number $ \sqrt{-1} $.
	Let the given RT-D1BCA be $ \mathcal{D} = (Q,\Sigma,\delta,q_{1},Q_{a}) $,
	where $ Q = \{q_{1}, \ldots, q_{n} \} $. We build $ \mathcal{M} = (Q',\Sigma,\Omega,\delta',q_{1,1},Q_{a}') $,
	where
	\begin{itemize}
		\item $ Q'= \{q_{j,1},\ldots,q_{j,n} \mid 1 \le j \le m\} $,
		\item $ Q'_{a} = \{ q_{m,i} \mid q_{i} \in Q_{a} \} $, and 
		\item $ \Omega = \{ \omega_{1}, \ldots, \omega_{n} \} $.
	\end{itemize}
	$ \mathcal{M} $ splits the computation into $m$ paths, i.e. $ \mathsf{path}_{j} $ ($ 1 \le j \le m $),
	with equal amplitude on the left end-marker $ \cent $. That is,
	\begin{equation}
		\delta' (q_{1,1},\cent) =
			\underbrace{\frac{1}{\sqrt{m}}(q_{1,t},0,\omega_{1})}_{\mathsf{path}_{1}} + \cdots +
			\underbrace{\frac{1}{\sqrt{m}}(q_{m,t},0,\omega_{1})}_{\mathsf{path}_{m}},
	\end{equation}
	whenever $ \delta(q_{1},\cent,q_{t},0) = 1 $, where $ 1 \le t \le n $.
	Until reading the right end-marker $ \dollar $, $ \mathsf{path}_{j} $ proceeds in the following way:
	For each  $ \sigma \in \Sigma $ and $  s \in \{1,\ldots,n\} $,
	\begin{eqnarray}
		\label{eq:deterministic-transition}
		\mathsf{path}_{j}: \delta'(q_{j,s},\sigma) & = & (q_{j,t},c_{j},\omega_{s})	
	\end{eqnarray}
	whenever $ \delta(q_{s},\sigma,q_{t},c) = 1 $, where $ 1 \le t \le n $, and
	\begin{itemize}
		\item $ c_{j} = j $ if $ c = 1 $,
		\item $ c_{j} = m-j+1 $ if $ c = -1 $, and
		\item $ c_{j} = 0 $, otherwise.
	\end{itemize} 
	
To paraphrase, each path separately simulates\endnote{Note that each transition of $ \mathcal{M} $ in Equation \ref{eq:deterministic-transition} writes a symbol determined by the source state of the corresponding transition of  $ \mathcal{D} $ to the register. This ensures the orthonormality condition for quantum machines described earlier.}
the computation of $ \mathcal{D} $ on the input string, going through states that correspond to the states of $ \mathcal{D} $, and incrementing their counters whenever $ \mathcal{D} $ changes its counter, as follows:
	\begin{itemize}
		\item $ \mathsf{path}_{j} $ increments the counter by $ j $
			whenever $ \mathcal{D} $ increments the counter by 1,
		\item $ \mathsf{path}_{j} $ increments the counter by $ m-j+1 $
			whenever $ \mathcal{D} $ decrements the counter by 1, and
		\item $ \mathsf{path}_{j} $ does not make any incrementation, otherwise.
	\end{itemize}

	On symbol $ \dollar $, the following transitions are executed (note that the counter updates in this last step are also made according to the setup described above):
	\newline
	If $ q_{t} \in Q_{a} $,
	\begin{equation}
		\label{eq:deterministic-accept}
		\mathsf{path}_{j}: \delta' (q_{j,s},\dollar) =
		\frac{1}{\sqrt{m}} \sum_{l=1}^{m} e^{\frac{2 \pi i}{m}jl} (q_{l,t},c_{j},\omega_{s})
	\end{equation}
	and if $ q_{t} \notin Q_{a} $,
	\begin{equation}
		\label{eq:deterministic-reject}
		\mathsf{path}_{j}: \delta' (q_{j,s},\dollar) = (q_{j,t},c_{j},\omega_{s}),
	\end{equation}
	whenever $ \delta(q_{s},\dollar,q_{t},c) = 1 $, where $ 1 \le t \le n $. 
	
The essential idea behind this setup, where different paths increment their counters with different values to represent increments and decrements performed by $ \mathcal{D} $ is that the increment values used by $ \mathcal{M} $ have been selected carefully to ensure that the counter will have the same value in all of $ \mathcal{M} $'s paths at any time if $ \mathcal{D} $'s counter is zero at that time. Furthermore, all of $ \mathcal{M} $'s paths are guaranteed to have different counter values if $ \mathcal{D} $'s counter is nonzero\endnote{This idea has been adapted from an algorithm by Kondacs and Watrous for a different type of quantum automaton, whose analysis can be found in \citep{KW97}.}.

\begin{figure}[h]
       \caption{$ N $-way quantum Fourier transform}
       \centering
       \fbox{
       \begin{minipage}{0.8\textwidth}
       \footnotesize
               Let $ N > 1 $ be a integer. The $ N $-way QFT is the transformation
\begin{equation}
                       \delta(d_{j}) \rightarrow \alpha
                               \sum\limits_{l = 1}^{N} e^{\frac{2 \pi i }{N}jl} (r_{l}), ~~~~ 1 \le j \le N ,
\end{equation}
       from the
               \textit{domain} states  $ d_{1}, \ldots, d_{N} $ to the
\textit{range} states $ r_{1}, \ldots, r_{N} $.
               $ r_{N} $ is the \textit{distinguished} range element. 
               $ \alpha $ is a real number such that $ \alpha^{2}N \leq 1 $.
               The QFT  can be used to check whether separate computational paths
of a quantum program that are in superposition have converged to the
same configuration at a particular step. Assume that the program has
previously split to $N$ paths, each of which have the same amplitude,
and whose state components are the  $ d_{j} $). In all the uses of the
QFT in our algorithms, one of the following conditions will be
satisfied:
\begin{enumerate}
                       \item The WOM component of the configuration is different in each
of the $N$ paths: In this case, the QFT will further divide each path
to $N$ subpaths, that will differ from each other by the internal
state component. No interference will take place.
                       \item Each path has the same WOM content at the moment of the QFT:
In this case, the paths that have $ r_{1}, \ldots, r_{N-1} $ as their
state components will destructively interfere with each
other \citep{YS09B}, and $ \alpha^{2} N $ of the probability of
the $N$ incoming paths will be accumulated on a single resulting path
with that WOM content, and $ r_{N} $ as its state component.
               \end{enumerate}
       \end{minipage}}
       \label{fig:wom:N-way-QFT}
\end{figure}

	For a given input string $ w \in \Sigma^{*} $,
	\begin{enumerate}
		\item if $ \mathcal{D} $ ends up in a state not in $ Q_{a} $ (and so $ w \notin L $), 
			then $ \mathcal{M} $ rejects the input in each of its $m$ paths, and the overall
			rejection probability is 1;
		\item if $ \mathcal{D} $ ends up in a state in $ Q_{a} $, all paths make an
			$ m $-way QFT (see Figure \ref{fig:wom:N-way-QFT}) 
			whose distinguished target is an accepting state:
			\begin{enumerate}
				\item if the counter of $ \mathcal{D} $ is zero (and so $ w \in L $),
					all paths have the same counter value, that is, they will interfere with each other, and so $ \mathcal{M} $ will accept with probability 1;
				\item if the counter of $ \mathcal{D} $ is not zero (and so $ w \notin L $), there will be no interference, and each path will end by accepting $ w $ with probability $ \frac{1}{m^{2}} $, leading to a total acceptance probability of $ \frac{1}{m} $, and a rejection probability of $ 1 - \frac{1}{m} $.
			\end{enumerate}
	\end{enumerate}
 \end{proof}

\begin{proof}[Proof of Theorem \ref{wom:thm:RT-Q1BCA-by-RT-QFA-IOC}]
	Given	a RT-Q1BCA that recognizes a language $L$ $ \mathcal{M} $ with error bound $ \epsilon < \frac{1}{2} $, we build a RT-QFA-IOC($ m $) $ \mathcal{M'} $, using essentially the same construction as in 
	Lemma \ref{wom:lem:RT-D1BCA-by-RT-QFA-IOC}: $ \mathcal{M'} $ simulates $m$ copies of $ \mathcal{M} $, and these copies use the set of increment sizes described in the proof of Lemma \ref{wom:lem:RT-D1BCA-by-RT-QFA-IOC} to mimic the updates to $ \mathcal{M} $'s counter. Unlike the deterministic machine of that lemma, $ \mathcal{M} $ can fork to multiple computational paths, which is handled by modifying the transformation of Equation \ref{eq:deterministic-transition} as 
	\begin{equation}
		\label{eq:quantum-transition}
		\mathsf{path}_{j}: \delta'(q_{j,s},\sigma,q_{j,t},c_{j},\omega) = \alpha
	\end{equation}
	whenever $ \delta(q_{s},\sigma,q_{t},c,\omega) = \alpha $, where $ 1 \le t \le n $, and
	$ \omega \in \Omega $, and that of Equation \ref{eq:deterministic-accept} as
	\begin{equation}
		\label{eq:quantum-accept}
		\mathsf{path}_{j}: \delta' (q_{j,s},\dollar,q_{l,t},c_{j},\omega) =
		\frac{\alpha}{\sqrt{m}} e^{\frac{2 \pi i}{m}jl},
		\mbox{ for } l \in \{1,\ldots,m\}
	\end{equation}
	whenever $ \delta(q_{s},\dollar,q_{t},c,\omega) = \alpha $, where $ 1 \le t \le n $ and $ \omega \in \Omega $;
causing the corresponding paths of the $m$ copies of $ \mathcal{M} $ to undergo the $m$-way QFTs associated by each accept state as described above at the end of the input.

We therefore have that the paths of $ \mathcal{M} $ that end in non-accept states do the same thing with the same total probability in $ \mathcal{M'} $.  The paths of $ \mathcal{M} $ that end in accept states with the counter containing zero make $ \mathcal{M'} $ accept also with their original total probability, thanks to the QFT. The only mismatch between the machines is in the remaining case of the paths of $ \mathcal{M} $ that end in accept states with a nonzero counter value. As explained in the proof of Lemma \ref{wom:lem:RT-D1BCA-by-RT-QFA-IOC}, each such path will contribute $ \frac{1}{m} $ of its probability to acceptance, and the rest to rejection.

	For any given input string $ w \in \Sigma^{*} $:
	\begin{itemize}
		\item If $ w \in L $, we have $ f_{\mathcal{M}}^{a}(w) \geq 1-\epsilon $ 
			and $ f_{\mathcal{M}}^{r}(w) \leq \epsilon $, then 
			\begin{equation}
				f_{\mathcal{M'}}^{a}(w) = f_{\mathcal{M}}^{a}(w) + \frac{1}{m} f_{\mathcal{M}}^{r}(w) 
				\geq 1 -\epsilon .
			\end{equation}
		\item If $ w \notin L $, we have $ f_{\mathcal{M}}^{a}(w) \leq \epsilon $ and 
			$ f_{\mathcal{M}}^{r}(w) \geq 1 - \epsilon $, then 
			\begin{equation}
				f_{\mathcal{M'}}^{a}(w) = f_{\mathcal{M}}^{a}(w) + \frac{1}{m} f_{\mathcal{M}}^{r}(w) 
				\leq  \epsilon + \frac{1}{m}(1-\epsilon).
			\end{equation}
	\end{itemize}
	Therefore, by setting $ m $ to a value greater than $ \frac{2-2\epsilon}{1-2\epsilon} $,
	$ L $ will be recognized by $ \mathcal{M'} $ with error bound 
	$ \epsilon' = \epsilon + \frac{1}{m}(1-\epsilon) < \frac{1}{2} $.
	Moreover, by setting $ m $ to sufficiently large values, $ \epsilon' $ can be tuned to be arbitrarily close to $ \epsilon $.
 \end{proof}

\begin{corollary}
	If $ L $ is recognized by a RT-Q1BCA (or a RT-P1BCA) $ \mathcal{P} $ 
	with negative one-sided error bound $ \epsilon < 1 $,
	then $ L $ is recognized by a RT-QFA-IOC $ \mathcal{M} $ with 
	negative one-sided error bound $ \epsilon' $, i.e. $ \epsilon < \epsilon' < 1 $.
	Moreover, $ \epsilon' $ can be tuned to be arbitrarily close to $ \epsilon $.
\end{corollary}

For a given nonnegative integer $ k $, $ L_{eq-k} $ is the language defined over the alphabet 
$ \{a_{1},\ldots,a_{k},b_{1},\ldots,b_{k}\} $ as the set of all strings containing equal
numbers of $ a_{i} $'s and $ b_{i} $'s, for each $ i \in \{1,\ldots,k\} $.

\begin{fact}
	\citep{Fr79}
	For any nonnegative $ k $, 
	$ L_{eq-k} $ can be recognized by a RT-P$ 1 $BCA with negative one-sided bounded error $ \epsilon $,
	where $ \epsilon < \frac{1}{2} $.	
\end{fact}

\begin{corollary}
	RT-QFA-IOCs can recognize some non-context-free languages with bounded error.
\end{corollary}

We have therefore established that realtime quantum finite automata equipped with a WOM tape are more powerful than plain RT-QFAs, even when the WOM in question is restricted to be just a counter.

$ L_{eq-1} $'s complement, which can of course be recognized with positive one-sided bounded error by a RT-QFA-IOC by the results above, is a deterministic context-free language (DCFL). Using the fact \citep{AGM92} that no nonregular DCFL can be recognized by a nondeterministic TM using $ o(\log(n)) $ space, together with Lemma
\ref{lem:classical-simulated-by-quantum}, we are able to conclude the following.

\begin{corollary}
	\label{cor:QTM-WOMS-superior-PTM-WOMs}
       QTM-WOMs are strictly superior to PTM-WOMs for any space bound $ o(\log(n)) $
       in terms of language recognition with positive one-sided bounded error.
\end{corollary}

\subsection*{Unbounded error} \label{wom:ioc:unbounded}

The simulation method introduced in Lemma \ref{wom:lem:RT-D1BCA-by-RT-QFA-IOC} turns out to be useful in the analysis of the power of increment-only counter machines in the unbounded error mode as well:

\begin{theorem}
	\label{wom:thm:RT-NQ1BCA-by-RT-QFA-IOC-one-sided}
	Any language recognized by a nondeterministic realtime automaton with one blind counter (RT-NQ$ 1 $BCA) is
	recognized by a RT-QFA-IOC with cutpoint $ \frac{1}{2} $.
\end{theorem}
\begin{proof}
	Given a RT-NQ$ 1 $BCA $ \mathcal{N} $, we note that it is just a RT-Q$ 1 $BCA recognizing a language $ L $ with positive one-sided unbounded error \citep{YS10A}, and we can simulate it using the technique described in the proof of Theorem \ref{wom:thm:RT-Q1BCA-by-RT-QFA-IOC}. We set $m$, the number of copies of the RT-Q$ 1 $BCA to be parallelly simulated, to 2. We obtain a RT-QFA-IOC(2) $ \mathcal{M} $ such that
	\begin{enumerate}
		\item paths of $ \mathcal{N} $ that end in an accepting state with the counter equaling zero
			lead  $ \mathcal{M} $ to accept with the same total probability;
		\item paths of $ \mathcal{N} $ that end in an accepting state with a nonzero counter value
			contribute half of their probability to $ \mathcal{M} $'s acceptance probability, with the other half contributing to rejection; and
		\item paths of $ \mathcal{N} $ that end in a reject state
			cause $ \mathcal{M} $ to reject with the same total probability.
	\end{enumerate}
Finally, we modify the transitions on the right end-marker that enter the reject states mentioned in the third case above, so that they are replaced by equiprobable transitions to an (accept,reject) pair of states.
	The resulting machine recognizes $ L $ with ``one-sided" cutpoint $ \frac{1}{2} $, that is, 
	the overall acceptance probability 
	exceeds $ \frac{1}{2} $ for the members of the language, and equals $ \frac{1}{2} $ for the nonmembers.
 \end{proof}

We now present a simulation of a classical model with non-blind counter.

\begin{theorem}
	\label{wom:1-rev-RT-D1CA-by-RT-QFA-IOC}
	If $ L $ is recognized by a realtime deterministic one-reversal one-counter automaton (1-rev-RT-D1CA), then it is recognized by a RT-QFA-IOC with cutpoint $ \frac{1}{2} $.
\end{theorem}
\begin{proof}
	We assume that the 1-rev-RT-D1CA $ \mathcal{D} = (Q,\Sigma,\delta,q_{1},Q_{a}) $ recognizing $ L $
	is in the following canonical form:
	\begin{itemize}		
		\item the counter value of $ \mathcal{D} $ never becomes nonnegative;
		\item the transition on $ \cent $ does not make any change 
			($ \delta(q_{1},\cent,0,q_{1})=1 $, and $ D_{c}(q_{1})=0 $);
		\item $ Q $ is the union of two disjoint subsets $ Q_{1} $ and $ Q_{2} $, i.e.
			\begin{enumerate}
				\item until the first decrement, the status of the counter is never checked
					-- this part is implemented by the members of $ Q_{1} $,
				\item during the first decrement, the internal state of $ \mathcal{D} $ switches to one of 
					the members of $ Q_{2} $, and
				\item the computation after the first decrement is implemented by the members of $ Q_{2} $;
			\end{enumerate}
		\item once the counter value is detected as zero, the status of the counter is not checked again.
	\end{itemize}	
	We will construct a RT-QFA-IOC $ \mathcal{M} = (Q',\Sigma,\Omega,\delta',q_{1},Q'_{a}) $, 
	to recognize $ L $ with cutpoint $ \frac{1}{2} $,
	where 
	\begin{itemize}
		\item $ Q' = \{q_{1}\} \cup \{q_{j,i} \mid j \in \{1,\ldots,4\}, i \in \{1,\ldots,|Q|\}  \}, $
		\item $ Q'_{a} = \{ q_{j,i} \mid j \in \{1,2,3\}, q_{i} \in Q_{a} \} \cup
			\{ q_{4,i} \mid q_{i} \in Q_{r} \} $, and
		\item $ \Omega =  \{  \omega_{i} \cup \omega'_{i} \mid i \in \{1,\ldots,|Q|\} \} $.
	\end{itemize}
and the details of $ \delta' $ are given in Figures \ref{wom:fig:transition-1-rev-RT-D1CA-by-RT-QFA-IOC-1}
	and \ref{wom:fig:transition-1-rev-RT-D1CA-by-RT-QFA-IOC-2}.
	
	\begin{figure}[h!]
	\caption{The details of the transition function of the RT-QFA-IOC presented in the proof of 
		Theorem \ref{wom:1-rev-RT-D1CA-by-RT-QFA-IOC} (I)}
	\centering
	\fbox{
	\begin{minipage}{0.95\textwidth}
		\footnotesize
		In the following, ``$ * $" means that the corresponding transition does not depend on 
		the status of the counter.	
		\\
		(i) On symbol $ \cent $:
		\begin{equation}
			\delta'(q_{1},\cent) = \underbrace{ \frac{1}{\sqrt{2}} (q_{1,1},0,\omega_{1}) }_{\mathsf{path}_{1}} + 
				\underbrace{ \frac{1}{\sqrt{2}} (q_{2,1},0,\omega_{1}) }_{\mathsf{path}_{2}} 
		\end{equation}
		(ii) On symbol $ \sigma \in \Sigma $: for each $ q_{i} \in Q_{1} $, 
		if $ \delta(q_{i},\sigma,*,q_{j}) = 1 $ and $ q_{j} \in Q_{1} $, then 
		\begin{equation}
			\begin{array}{lcl}
				\mathsf{path}_{1}: \delta'(q_{1,i},\sigma) & = & 
					\underbrace{(q_{1,j},D_{c}(q_{j}),\omega_{i})}_{\mathsf{path}_{1}}
					\\
				\mathsf{path}_{2}: \delta'(q_{2,i},\sigma) & = & 
					\underbrace{(q_{2,j},0,\omega_{i})}_{\mathsf{path}_{2}}		
			\end{array}.
		\end{equation}
		(iii) On symbol $ \dollar $: for each $ q_{i} \in Q_{1} $, 
			if $ \delta(q_{i},\dollar,*,q_{j}) = 1 $ and $ q_{j} \in Q_{1} $, then
		\begin{equation}
			\begin{array}{lcl}
				\mathsf{path}_{1}: \delta'(q_{1,i},\sigma) & = & 
					\underbrace{(q_{1,j},0,\omega_{i})}_{\mathsf{path}_{1}}  \\
				\mathsf{path}_{2}: \delta'(q_{2,i},\sigma) & = & 
					\underbrace{(q_{2,j},0,\omega_{i})}_{\mathsf{path}_{2}}
			\end{array}.
		\end{equation}
		\end{minipage}
		}
		\label{wom:fig:transition-1-rev-RT-D1CA-by-RT-QFA-IOC-1}
	\end{figure}
	
	\begin{figure}[h!]
	\caption{The details of the transition function of the RT-QFA-IOC presented in the proof of 
		Theorem \ref{wom:1-rev-RT-D1CA-by-RT-QFA-IOC} (II)}
	\centering
	\fbox{
	\begin{minipage}{0.95\textwidth}
		\footnotesize		
		(iv) On symbol $ \sigma \in \Sigma $: for each $ q_{i} \in Q $, 
		if $ \delta(q_{i},\sigma,1,q_{j}) = 1 $ and $ q_{j} \in Q_{2} $, then 
		\begin{equation}
			\begin{array}{lcl}
				\mathsf{path}_{1}: \delta'(q_{1,i},\sigma) & = & 
					\underbrace{\frac{1}{\sqrt{3}}(q_{1,j},0,\omega_{i})}_{\mathsf{path}_{1}} +
					\underbrace{\frac{1}{\sqrt{3}}(q_{3,j},0,\omega'_{i})}_{\mathsf{path}_{3}} + 
					\underbrace{\frac{1}{\sqrt{3}}(q_{4,j},0,\omega'_{i})}_{\mathsf{path}_{4}}  \\
				\mathsf{path}_{2}: \delta'(q_{2,i},\sigma) & = & 
					\underbrace{\frac{1}{\sqrt{3}}(q_{2,j},c_{2},\omega_{i})}_{\mathsf{path}_{2}} +
					\underbrace{\frac{1}{\sqrt{3}}(q_{3,j},c_{2},\omega'_{i})}_{\mathsf{path}_{3}} - 
					\underbrace{\frac{1}{\sqrt{3}}(q_{4,j},c_{2},\omega'_{i})}_{\mathsf{path}_{4}} 
			\end{array},
		\end{equation}
		and if $ \delta(q_{i},\sigma,0,q_{j}) = 1 $, then
		\begin{equation}
			\begin{array}{lcl}
				\mathsf{path}_{3}: \delta'(q_{3,i},\sigma) & = & 
					\underbrace{(q_{3,j},0,\omega_{i})}_{\mathsf{path}_{3}} \\
				\mathsf{path}_{4}: \delta'(q_{4,i},\sigma) & = & 
					\underbrace{(q_{4,j},0,\omega_{i})}_{\mathsf{path}_{4}} \\
			\end{array},
		\end{equation}
		where $ c_{2}=1 $ only if $ D_{c}(q_{j})=-1 $.
		\\
		(iii) On symbol $ \dollar $: for each $ q_{i} \in Q $, if $ \delta(q_{i},\dollar,1,q_{j}) = 1 $
		and $ q_{j} \in Q_{2} $, then
		\begin{equation}
			\begin{array}{lcl}
				\mathsf{path}_{1}: \delta'(q_{1,i},\sigma) & = & 
					\underbrace{\frac{1}{\sqrt{3}}(q_{1,j},0,\omega_{i})}_{\mathsf{path}_{1}} +
					\underbrace{\frac{1}{\sqrt{3}}(q_{3,j},0,\omega_{i})}_{\mathsf{path}_{3}} + 
					\underbrace{\frac{1}{\sqrt{3}}(q_{4,j},0,\omega_{i})}_{\mathsf{path}_{4}}  \\
				\mathsf{path}_{2}: \delta'(q_{2,i},\sigma) & = & 
					\underbrace{\frac{1}{\sqrt{3}}(q_{2,j},c_{2},\omega_{i})}_{\mathsf{path}_{2}} +
					\underbrace{\frac{1}{\sqrt{3}}(q_{3,j},c_{2},\omega_{i})}_{\mathsf{path}_{3}} - 
					\underbrace{\frac{1}{\sqrt{3}}(q_{4,j},c_{2},\omega_{i})}_{\mathsf{path}_{4}} 
			\end{array},
		\end{equation}		
		and if $ \delta(q_{i},\dollar,0,q_{j}) = 1 $, then
		\begin{equation}
			\begin{array}{lcl}
				\mathsf{path}_{3}: \delta'(q_{3,i},\sigma) & = & 
					\underbrace{(q_{3,j},0,\omega_{i})}_{\mathsf{path}_{3}} \\
					\mathsf{path}_{4}: \delta'(q_{4,i},\sigma) & = & 
						\underbrace{(q_{4,j},0,\omega_{i})}_{\mathsf{path}_{4}} \\
				\end{array},
			\end{equation}
			where $ c_{2}=1 $ only if $ D_{c}(q_{j})=-1 $.
		\end{minipage}
		}
		\label{wom:fig:transition-1-rev-RT-D1CA-by-RT-QFA-IOC-2}
	\end{figure}
	
	$ \mathcal{M} $ starts by branching to two paths, $ \mathsf{path}_{1} $ and  $ \mathsf{path}_{2} $, 
	with equal amplitude. These paths simulate 
	$ \mathcal{D} $ in parallel according to the specifications in Figure 
	\ref{wom:fig:transition-1-rev-RT-D1CA-by-RT-QFA-IOC-1} until $ \mathcal{D} $ decrements its counter 
	for the first time. From that step on, $ \mathsf{path}_{1} $ and  $ \mathsf{path}_{2} $ split further 
	to create new offshoots (called $ \mathsf{path}_{3} $ and  $ \mathsf{path}_{4} $,) on every symbol 
	until the end of the computation, as seen in Figure \ref{wom:fig:transition-1-rev-RT-D1CA-by-RT-QFA-IOC-2}.
	Throughout the computation, $ \mathsf{path}_{1} $ (resp., $ \mathsf{path}_{2} $) increments its counter 
	whenever $ \mathcal{D} $ is supposed to increment (resp., decrement) its counter. Since $ \mathcal{M} $'s 
	counter is write-only, it has no way of determining which transition $ \mathcal{D} $ will make depending on its 
	counter sign. This problem is solved by assigning different paths of $ \mathcal{M} $ to these branchings of 
	$ \mathcal{D} $:  $ \mathsf{path}_{1} $ and  $ \mathsf{path}_{2} $ (the ``pre-zero paths") always assume that $ \mathcal{D} $'s counter 
	has not returned to zero yet by being decremented, whereas $ \mathsf{path}_{3} $s and  $ \mathsf{path}_{4} $s (the ``post-zero paths") 
	carry out their simulations by assuming otherwise. Except for $ \mathsf{path}_{4} $s, all paths imitate $ 
	\mathcal{D} $'s decision at the end of the computation. $ \mathsf{path}_{4} $s, on the other hand, accept if 
	and only if their simulation of $ \mathcal{D} $ rejects the input.

	If $ \mathcal{D} $ never decrements its counter, 
	$ \mathcal{M} $ ends up with the same decision as $ \mathcal{D} $ with probability 1.
	We now focus on the other cases.
As seen in Figure \ref{wom:fig:transition-1-rev-RT-D1CA-by-RT-QFA-IOC-2}, the pre-zero paths lose some of their amplitude on each symbol in this stage by performing a QFT to a new pair of post-zero paths. The outcome of this transformation depends on the status of $ \mathcal{D} $'s counter at this point in the simulation by the pre-zero paths:
	\begin{itemize}
		\item If $ \mathcal{D} $'s counter has not yet returned to zero, then  $ \mathsf{path}_{2} $'s counter has a smaller value than $ \mathsf{path}_{1} $'s counter, and so they cannot interfere via the QFT.
			The newly created post-zero paths will contribute equal amounts to the acceptance and rejection probabilities at the end of the computation.
		\item If $ \mathsf{path}_{1} $ and $ \mathsf{path}_{2} $ have the same counter value as a result of this transition, this indicates that $ \mathcal{D} $ has performed exactly as many decrements as its previous increments, and its counter is therefore zero. The paths interfere, the target $ \mathsf{path}_{4} $'s cancel each other, and $ \mathsf{path}_{3} $ survives after the QFT with a probability that is twice
			that of the total probability of the ongoing pre-zero paths. 			
\end{itemize}
	As a result, it is guaranteed that the path that is carrying out the correct simulation of $ \mathcal{D} $ will dominate $ \mathcal{M} $'s decision at the end of the computation: If $ \mathcal{D} $'s counter ever returns to zero, the $ \mathsf{path}_{3} $ that is created at the moment of that last decrement will have sufficient probability to tip the accept/reject balance. If $ \mathcal{D} $'s counter never returns to zero, then the common decision by the pre-zero paths on the right end-marker will determine whether the overall acceptance or the rejection probability will be greater than $ \frac{1}{2} $.
 \end{proof}

Consider the following language \citep{NH71}:

\begin{equation}
	\footnotesize
	L_{NH} = \{a^{x}ba^{y_{1}}ba^{y_{2}}b \cdots a^{y_{t}}b \mid x,t,y_{1}, \cdots, y_{t}
		\in \mathbb{Z}^{+} \mbox{ and } \exists k ~ (1 \le k \le t), x=\sum_{i=1}^{k}y_{i} \}
\end{equation}

$ L_{NH} $ is recognizable by both 1-rev-RT-D1CAs and RT-N1BCAs\endnote{RT-N1BCAs can also recognize
$ L_{center} = \{ ubv \mid u,v \in \{a,b\}^{*},|u|=|v| \} $,
and the languages studied in \citep{FYS10A}, none of which can be recognized by RT-QFAs with unbounded error.} (and so RT-NQ1BCAs).
It is known \citep{NH71,FK94,LQ08,YS10C} that neither a RT-QFA nor a 
$o(\log(\log(n)))$-space PTM can recognize $ L_{NH} $ with unbounded error. 
We therefore have the following corollary.

\begin{corollary}
	\label{wom:cor:QTM-WOM-superior-PTM-WOM}
       QTM-WOMs are strictly superior to PTM-WOMs for any space bound $ o(\log(\log(n))) $
       in terms of language recognition with unbounded error.
\end{corollary}

\section*{Machines with push-only stack}

We conjecture that allowing more than one nonblank/nonempty symbol in
the WOM tape alphabet of a QFA increases its computational power. We
consider, in particular, the language $ L_{twin}=\{wcw \mid w \in
\{a,b\}^{*} \} $:

\begin{theorem}
       \label{wom:thm:Ltwin-by-RT-QFA-POS}
       There exists a RT-QFA-POS that recognizes the language $ L_{twin} $
       with negative one-sided error bound $ \frac{1}{2} $.
\end{theorem}
\begin{proof}
       We construct a RT-QFA-POS $
\mathcal{M}=(Q,\Sigma,\Gamma,\Omega,\delta,q_{1},Q_{a}) $,
       where $ Q=\{q_{1},q_{2},q_{3},p_{1},p_{2},p_{3}\} $, $ Q_{a} = \{q_{2}\} $,
       $ \Omega=\{ \omega_{1},\omega_{2} \} $, and $ \Gamma=\{\#,a,b,\varepsilon\} $.
       The transition details are shown in Figure \ref{wom:fiq:Ltwin}.
       
       \begin{figure}[h!]
               \caption{The transitions of the RT-QFA-POS of Theorem
\ref{wom:thm:Ltwin-by-RT-QFA-POS}}
               \centering
               \fbox{
               \begin{minipage}{0.85\textwidth}
               \small
               On symbol $ \cent $:
               \begin{equation}
                       \delta(q_{1},\cent) = \underbrace{ \frac{1}{\sqrt{2}}
(q_{1},\varepsilon,\omega_{1}) }_{\mathsf{path}_{1}} +
                               \underbrace{\frac{1}{\sqrt{2}} (p_{1},\varepsilon,\omega_{1})}_{\mathsf{path}_{2}}
               \end{equation}
               On symbols from $ \Sigma $:
               \begin{equation}
                       \mathsf{path}_{1}: \left\{
                               \begin{array}{lcl}
                                       \delta(q_{1},a) & = & (q_{1},a,\omega_{1}) \\
                                       \delta(q_{2},a) & = & (q_{2},\varepsilon,\omega_{1}) \\
                                       \delta(q_{1},b) & = & (q_{1},b,\omega_{1}) \\
                                       \delta(q_{2},b) & = & (q_{2},\varepsilon,\omega_{1}) \\
                                       \delta(q_{1},c) & = & (q_{2},\varepsilon,\omega_{1}) \\
                                       \delta(q_{2},c) & = & (q_{3},\varepsilon,\omega_{1}) \\
                                       \delta(q_{3},a) & = & (q_{3},\varepsilon,\omega_{2}) \\
                                       \delta(q_{3},b) & = & (q_{3},\varepsilon,\omega_{2}) \\
                                       \delta(q_{3},c) & = & (q_{3},\varepsilon,\omega_{2})
                               \end{array}
                       \right.
               \end{equation}
               \begin{equation}
                       \mathsf{path}_{2}: \left\{
                               \begin{array}{lcl}
                                       \delta(p_{1},a) & = & (p_{1},\varepsilon,\omega_{1}) \\
                                       \delta(p_{2},a) & = & (p_{2},a,\omega_{1}) \\
                                       \delta(p_{1},b) & = & (p_{1},\varepsilon,\omega_{1}) \\
                                       \delta(p_{2},b) & = & (p_{2},b,\omega_{1}) \\
                                       \delta(p_{1},c) & = & (p_{2},\varepsilon,\omega_{1}) \\
                                       \delta(p_{2},c) & = & (p_{3},\varepsilon,\omega_{1}) \\
                                       \delta(p_{3},a) & = & (p_{3},\varepsilon,\omega_{2}) \\
                                       \delta(p_{3},b) & = & (p_{3},\varepsilon,\omega_{2}) \\
                                       \delta(p_{3},c) & = & (p_{3},\varepsilon,\omega_{2})                                                                              
                               \end{array}
                       \right.
               \end{equation}
               On symbol $ \dollar $:
               \begin{equation}
                       \mathsf{path}_{1}: \left\{
                               \begin{array}{lcl}
                                       \delta(q_{1},\dollar) & = & (q_{1},\varepsilon,\omega_{1}) \\
                                       \delta(q_{2},\dollar) & = & \frac{1}{\sqrt{2}} (q_{2},\varepsilon,\omega_{1}) +
                                               \frac{1}{\sqrt{2}} (q_{3},\varepsilon,\omega_{2}) \\
                                       \delta(q_{3},\dollar) & = & (q_{3},\varepsilon,\omega_{1})
                               \end{array}
                       \right.
               \end{equation}
               \begin{equation}
                       \mathsf{path}_{2}: \left\{
                               \begin{array}{lcl}
                                       \delta(p_{1},\dollar) & = & (p_{1},\varepsilon,\omega_{1}) \\
                                       \delta(p_{2},\dollar) & = & \frac{1}{\sqrt{2}} (q_{2},\varepsilon,\omega_{1}) -
                                               \frac{1}{\sqrt{2}} (q_{3},\varepsilon,\omega_{2}) \\
                                        \delta(p_{3},\dollar) & = & (q_{3},\varepsilon,\omega_{1})
                               \end{array}
                       \right.
               \end{equation}
       \end{minipage}
       }
       \label{wom:fiq:Ltwin}
       \end{figure}

       \begin{enumerate}
               \item The computation splits into two paths, $ \mathsf{path}_{1} $
and $ \mathsf{path}_{2} $, with equal
                       amplitude at the beginning.
               \item $ \mathsf{path}_{1} $ (resp., $ \mathsf{path}_{2} $) scans the
input, and
                       copies $ w_{1} $ (resp., $ w_{2} $) to the POS if the input is of
the form $ w_{1} c w_{2} $,
                       where $ w_{1},w_{2} \in \{a,b\}^{*} $.
                       \begin{enumerate}
                               \item If the input is not of the form $ w_{1} c w_{2} $, both paths reject.
                               \item Otherwise, $ \mathsf{path}_{1} $
                                       and $ \mathsf{path}_{2} $ perform a QFT at the end of the
computation, where the distinguished range element is an accept state.
               \end{enumerate}
\end{enumerate}
The configurations at the ends of $ \mathsf{path}_{1} $ and $
\mathsf{path}_{2} $ interfere with each other,
i.e., the machine accepts with probability $ 1 $, if and only if the
input is of the form
$ w c w $, $ w \in \{a,b\}^{*} $.
Otherwise, each of $ \mathsf{path}_{1} $ and $ \mathsf{path}_{2} $
contributes at most $ \frac{1}{4} $ to the overall acceptance probability,
and the machine accepts with probability at most $ \frac{1}{2} $.
\end{proof}

\begin{lemma}
	\label{wom:lem:no-ptm-wom-L-twin}
        No PTM (or PTM-WOM) using $ o(\log(n)) $ space can recognize $
L_{twin} $ with bounded error.
\end{lemma}
\begin{proof}
       Any PTM using $ o(\log(n)) $ space to recognize $ L_{twin} $ with bounded error
       can be used to construct a PTM recognizing the palindrome language
       $ L_{pal} $ with bounded error using the same amount of space.
       (One would only need to modify the $ L_{twin} $ machine to treat the
right end-marker on the tape as
       the symbol $ c $, and switch its head direction when it attempts to
go past that symbol.)
       It is however known \citep{FK94} that no PTM using $ o(\log(n)) $
space can recognize $ L_{pal} $
       with bounded error.
\end{proof}

We are now able to state a stronger form of Corollary \ref{cor:QTM-WOMS-superior-PTM-WOMs},
which referred only to one-sided error:

\begin{corollary}
	\label{cor:QTM-WOMS-superior-PTM-WOMs-bounded-error}
       QTM-WOMs are strictly superior to PTM-WOMs for any space bound $ o(log(n)) $
       in terms of language recognition with bounded error.
\end{corollary}

\begin{openproblem}
       Can a probabilistic pushdown automaton recognize $ L_{twin} $ with
bounded error?
\end{openproblem}

\section*{Machines using two-way WOM tape}

In this section, we present a bounded-error RT-QFA-WOM that recognizes a language for which we currently do not know a RT-QFA-POS algorithm, namely,
\begin{equation}
	L_{rev} = \{wcw^{\rev}\ \mid w \in \{a,b\}^{*} \},
\end{equation}
where $w^{\rev}$ is the reverse of string $w$.
Note that this language can also be recognized by a deterministic pushdown automaton.

\begin{theorem}
	There exists a RT-QFA-WOM that recognizes $ L_{rev} $
	with negative one-sided error bound $ \frac{1}{2} $.
\end{theorem}
\begin{proof}
	(sketch)
	We will use almost the same technique presented in the proof of Theorem \ref{wom:thm:Ltwin-by-RT-QFA-POS}.
	The computation is split into two paths ($ \mathsf{path}_{1} $ and $ \mathsf{path}_{2} $) 
	with equal amplitude at the beginning of the computation.
	Each path checks whether the input string is of the form $ w_{1}cw_{2} $, where $ w_{1},w_{2} \in \{a,b\}^{*} $
	and rejects with probability 1 if it is not.
	We assume that the input string is of the form $ w_{1}cw_{2} $ in the rest of this proof.
	Until the $ c $ is read, $ \mathsf{path}_{1} $ copies $ w_{1} $ to the WOM tape, and $ \mathsf{path_{2}} $
	just moves the WOM tape head one square to the right at each step, without writing anything.
	After reading the $ c $, the direction of the WOM tape head is reversed in both paths.
	That is, $ \mathsf{path}_{1} $ moves the WOM tape head one square to the left at each step, without writing anything, while $ \mathsf{path_{2}} $ writes $ w_{2} $ in the reverse direction (from the right to the left) 
	on the WOM tape.
	When the right end-marker is read,
	the paths make a QFT, as in the proof of Theorem \ref{wom:thm:Ltwin-by-RT-QFA-POS}.
	It is easy to see that the two paths interfere if and only if $ w_{1} = w_{2}^{\rev} $, and the input string is accepted with probability 1 if it is a member of $ L_{rev} $, and with probability $ \frac{1}{2} $ otherwise.

 \end{proof}

By an argument similar to the one used in the proof of Lemma \ref{wom:lem:no-ptm-wom-L-twin}, 
$ L_{rev} $ can not be recognized with bounded error by any PTM using $ o(\log(n)) $ space,
since the existence of any such machine would lead to a PTM that recognizes the palindrome
language using the same amount of space.

\begin{openproblem}
	Can a RT-QFA-POS recognize $ L_{rev} $ with bounded error?
\end{openproblem}

\section*{Small amounts of WOM can be useful}

It is easy to see that a WOM of constant size adds no power to a conventional machine. All the algorithms we considered until now used $ \Omega(n) $ squares of the WOM tape on worst-case inputs. What is the minimum amount of WOM that is required by a QFA-WOM recognizing a nonregular language? Somewhat less ambitiously, one can ask whether there is any nonregular language recognized by a RT-QFA-WOM with sublinear space. We answer this question positively for middle-space usage \citep{Sz94}, 
that is, when we are only concerned with the space used by the machine when the input is a member of the language.

Let $ (i)^{\rev}_{2} $ be the reverse of the binary representation of $ i \in \mathbb{N} $. Consider the language
\begin{equation}
	L_{rev-bins}=\{ a (0)_{2}^{\rev} a (1)_{2}^{\rev} a \cdots a (k)_{2}^{\rev} a \mid k \in \mathbb{Z}^{+} \}.
\end{equation}
\begin{theorem}
	$ L_{rev-bins} $ can be recognized by a RT-QFA-WOM $ \mathcal{M} $ with negative one-sided error bound 
	$ \frac{3}{4} $, and the WOM usage of $ \mathcal{M} $ for the members of $ L_{rev-bins} $ is $ O(\log n) $,
	where $ n $ is the length of the input string.
\end{theorem}
\begin{proof}
	It is not hard to modify the RT-QFA-POS recognizing $ L_{twin} $ to obtain a new RT-QFA-POS, say $ \mathcal{M}' $, 
	in order to recognize language
	$ L_{twin'}=\{ (i)^{\rev}_{2} a (i+1)_{2}^{\rev} \mid i \geq 0 \} $ 
	with negative one-sided error bound $ \frac{1}{2} $.
	Our construction of $ \mathcal{M} $ will be based on $ \mathcal{M'} $.
	The main idea is to use $ \mathcal{M}' $ in a loop in order to check the consecutive
	blocks of $ \{0,1\}^{+}a\{0,1\}^{+} $ between two $ a $'s. In each iteration,
	the WOM tape head reverses direction, and so the previously used space can be used again and again.
	Note that, whenever $ \mathcal{M}' $ executes a rejecting transition, $ \mathcal{M} $ enters a path which will reject the input when it arrives at the right end-marker, and whenever $ \mathcal{M}' $ is supposed to execute an accepting transition 
	(except at the end of the computation), $ \mathcal{M} $ enters the next iteration.
	At the end of the input, the input is accepted by $ \mathcal{M} $
	if $ \mathcal{M}' $ accepts in its last iteration.
	
	Let $ w  $ be an input string.
	We assume that $ w $ is of the form
	\begin{equation}
		a \{0,1\}^{+} a \{0,1\}^{+} a \cdots a \{0,1\}^{+}a .
	\end{equation} 
	(Otherwise, it is rejected with probability 1.)
	At the beginning, the computation is split equiprobably into two branches,
	$ \mathsf{branch}_{1} $ and $ \mathsf{branch}_{2} $.
	(These will never interfere with each other.)
	$ \mathsf{branch}_{1} $ (resp., $ \mathsf{branch}_{2} $) enters the block-checking loop after reading the
	first (resp., the second) $ a $. Thus, at the end of the computation, 
	one of the branches is in the middle of an iteration, and the other one has just finished its final iteration.
	The branch whose iteration is interrupted by reading the end-marker accepts with probability 1.
	
	If $ w \in L_{rev-bins} $, neither branch enters a reject state, and the input is accepted with probability 1. 
	On the other hand, if $ w \notin L_{rev-bins}  $, there must be at least one block $ \{0,1\}^{+}a\{0,1\}^{+} $
	that is not a member of $ L_{twin'} $, and so the input is rejected with probability $ \frac{1}{2} $
	in one branch. 
	Therefore, the overall accepting probability can be at most $ \frac{3}{4} $.

It is easy to see that the WOM usage of this algorithm for members of $ L_{rev-bins} $ is $ O(\log n) $.
 \end{proof}

\section*{Conclusion}

In this paper, we showed that write-only memory devices can increase
the computational power of quantum computers. We considered quantum
finite automata augmented with WOMs, and  demonstrated several example
languages which are known to be unrecognizable by conventional
quantum computers with certain restrictions, but are recognizable by a
quantum computer employing a WOM under the same restrictions. The
QFA-WOM models under consideration were also shown to be able to
simulate certain classical machines that employ linear amounts of memory, and
are therefore much more powerful than finite automata. We also showed
that merely logarithmic amounts of WOM can be useful in the sense of
enabling the recognition of nonregular languages.

A close examination of our algorithms reveals that quantum computers using WOM are able to avoid 
the argument in the proof of Lemma \ref{wom:lem:classical-wom} thanks to 
their use of negative (and complex) transition amplitudes in the QFT, 
which enables two configurations with the same WOM value to cancel each other altogether, 
when they have suitable amplitudes.

If one changes the RT-QFA-POS model so that the POS is now an output tape, 
the machine described in Theorem \ref{wom:thm:Ltwin-by-RT-QFA-POS} 
becomes a realtime quantum finite state transducer (RT-QFST) 
\citep{FW01} computing the function \citep{SY10A}
\begin{equation} 
	f(x) = \left\lbrace 
		\begin{array}{ll}
			w, & \mbox{ if } x=wcw, \mbox{ where } w \in \{a,b\}^{*} \\
			\mbox{undefined}, & \mbox{ otherwise}
		\end{array},
	\right.
\end{equation}
with bounded error.
The arguments leading to Corollary \ref{cor:QTM-WOMS-superior-PTM-WOMs-bounded-error} can then
be rephrased in a straightforward way to show that conventional QTMs 
are strictly superior to PTMs in function computation for any common
space bound that is $ o(\log(n)) $.

Finally, assume that we make another change to the RT-QFST described
in the paragraph above, so that it prints the symbol $c$ when it is
about to accept. The resulting constant-space QTM is easily seen to be
computing a reduction from $ L_{twin} $ to the language $L_{1} = \{
\{a,b\}^{*}c \} $ with bounded error. But no PTM $\mathcal{P}$ using $
o(\log n) $ space can compute this reduction, since we could otherwise
build a PTM for deciding $ L_{twin} $ with the same error bound by
composing $\mathcal{P}$ with the finite automaton recognizing $L_{1}$.
The detailed examination of how and to what extent quantum reductions
outperform probabilistic reductions with common restrictions is an
interesting topic.

Note that it is already known that adding a WOM to a reversible
classical computer may increase its computational power, since
it enables one to embed irreversible tasks into
``larger" reversible tasks by using the WOM as a trashcan.
As a simple example, reversible finite automata (RFAs) can recognize
a proper subset of regular languages \citep{Pi87}, but RFA's with WOM
can recognize exactly the
regular languages, and nothing more.
In the quantum case, WOM can also have a similar effect.
For example, the computational power of the most restricted type of
quantum finite automata (MCQFAs) \citep{MC00}
is equal to RFAs, but it has been shown \citep{Pa00,Ci01} that
MCQFAs with WOM can recognize all and only the regular languages,
attaining the power of the most general quantum finite automata (QFA)
without WOM. In all these examples, the addition of WOM to a
specifically weak model raises it to the level of the most general
classical (deterministic) automaton.
On the other hand, in this work, we show that adding WOM to the most
general type of QFA results in a much more powerful model that can
achieve a task that is impossible for all sublogarithmic space PTMs.

Some remaining open problems related to this study can be listed as follows:
\begin{enumerate}
      \item Does a WOM add any power to quantum computers which are allowed to
              operate at logarithmic or even greater space bounds?
      \item How would having several separate WOMs, each of which would
              contain different strings, affect the performance?
      \item Is there a nontrivial lower bound to the amount of WOM
that is useful for the recognition of nonregular languages by
QFA-WOMs?
\end{enumerate}

\section*{Acknowledgements}

Yakary{\i}lmaz and Say were partially supported by the Scientific and Technological Research Council of
Turkey (T\"{U}B\.ITAK) with grant 108142. 
Freivalds and Agadzanyan  were partially supported by Grant No. 09.1570 from the
Latvian Council of Science and by Project 2009/0216/1DP/1.1.2.1.2/09/IPIA
/VIA/004 from the European Social Fund.

\theendnotes

\bibliographystyle{abbrvnat}
\bibliography{YakaryilmazSay}

\begin{thebibliography}{36}
\providecommand{\natexlab}[1]{#1}
\providecommand{\url}[1]{\texttt{#1}}
\expandafter\ifx\csname urlstyle\endcsname\relax
  \providecommand{\doi}[1]{doi: #1}\else
  \providecommand{\doi}{doi: \begingroup \urlstyle{rm}\Url}\fi

\bibitem[Aharonov et~al.(1998)Aharonov, Kitaev, and Nisan]{AKN98}
D.~Aharonov, A.~Kitaev, and N.~Nisan.
\newblock Quantum circuits with mixed states.
\newblock In \emph{STOC'98: Proceedings of the Thirtieth Annual ACM Symposium
  on Theory of Computing}, pages 20--30, 1998.

\bibitem[Alt et~al.(1992)Alt, Geffert, and Mehlhorn]{AGM92}
H.~Alt, V.~Geffert, and K.~Mehlhorn.
\newblock A lower bound for the nondeterministic space complexity of
  context-free recognition.
\newblock \emph{Information Processing Letters}, 42\penalty0 (1):\penalty0
  25--27, 1992.

\bibitem[Ambainis and Yakary{\i}lmaz(2010)]{AY10A}
A.~Ambainis and A.~Yakary{\i}lmaz.
\newblock \emph{Automata: from Mathematics to Applications}, chapter Automata
  and quantum computing.
\newblock 2010.
\newblock (In preparation).

\bibitem[Arora and Barak(2009)]{AB09}
S.~Arora and B.~Barak.
\newblock \emph{Computational Complexity: A Modern Approach}.
\newblock Cambridge University Press, 2009.

\bibitem[Blondel et~al.(2005)Blondel, Jeandel, Koiran, and Portier]{BJKP05}
V.~D. Blondel, E.~Jeandel, P.~Koiran, and N.~Portier.
\newblock Decidable and undecidable problems about quantum automata.
\newblock \emph{SIAM Journal on Computing}, 34\penalty0 (6):\penalty0
  1464--1473, 2005.
\newblock ISSN 0097-5397.

\bibitem[Bonner et~al.(2001)Bonner, Freivalds, and Kravtsev]{BFK01}
R.~Bonner, R.~Freivalds, and M.~Kravtsev.
\newblock Quantum versus probabilistic one-way finite automata with counter.
\newblock In \emph{SOFSEM 2007: Theory and Practice of Computer Science},
  volume 2234 of \emph{Lecture Notes in Computer Science}, pages 181--190,
  2001.

\bibitem[Bozapalidis(2003)]{Bo03}
S.~Bozapalidis.
\newblock Extending stochastic and quantum functions.
\newblock \emph{Theory of Computing Systems}, 36\penalty0 (2):\penalty0
  183--197, 2003.

\bibitem[Chan(1981)]{Ch81}
T.~Chan.
\newblock Reversal complexity of counter machines.
\newblock In \emph{STOC'81: Proceedings of the thirteenth annual ACM symposium
  on Theory of computing}, pages 146--157, 1981.

\bibitem[Ciamarra(2001)]{Ci01}
M.~P. Ciamarra.
\newblock Quantum reversibility and a new model of quantum automaton.
\newblock In \emph{FCT'01: Proceedings of the 13th International Symposium on
  Fundamentals of Computation Theory}, pages 376--379. Springer-Verlag, 2001.

\bibitem[Freivalds(1979)]{Fr79}
R.~Freivalds.
\newblock Fast probabilistic algorithms.
\newblock In \emph{Mathematical Foundations of Computer Science 1979},
  volume~74 of \emph{Lecture Notes in Computer Science}, pages 57--69, 1979.

\bibitem[Freivalds and Karpinski(1994)]{FK94}
R.~Freivalds and M.~Karpinski.
\newblock Lower space bounds for randomized computation.
\newblock In \emph{ICALP'94: Proceedings of the 21st International Colloquium
  on Automata, Languages and Programming}, pages 580--592, 1994.

\bibitem[Freivalds and Winter(2001)]{FW01}
R.~Freivalds and A.~Winter.
\newblock Quantum finite state transducers.
\newblock In \emph{SOFSEM 2001: Theory and Practice of Informatics}, pages
  233--242, 2001.

\bibitem[Freivalds et~al.(2010)Freivalds, Yakary{\i}lmaz, and Say]{FYS10A}
R.~Freivalds, A.~Yakary{\i}lmaz, and A.~C.~C. Say.
\newblock A new family of nonstochastic languages.
\newblock \emph{Information Processing Letters}, 110\penalty0 (10):\penalty0
  410--413, 2010.

\bibitem[Hirvensalo(2008)]{Hi08}
M.~Hirvensalo.
\newblock Various aspects of finite quantum automata.
\newblock In \emph{DLT'08: Proceedings of the 12th international conference on
  Developments in Language Theory}, pages 21--33, 2008.

\bibitem[Jeandel(2007)]{Je07}
E.~Jeandel.
\newblock Topological automata.
\newblock \emph{Theory of Computing Systems}, 40\penalty0 (4):\penalty0
  397--407, 2007.
\newblock ISSN 1432-4350.

\bibitem[Kondacs and Watrous(1997)]{KW97}
A.~Kondacs and J.~Watrous.
\newblock On the power of quantum finite state automata.
\newblock In \emph{FOCS'97: Proceedings of the 38th Annual Symposium on
  Foundations of Computer Science}, pages 66--75, 1997.

\bibitem[Kravtsev(1999)]{Kr99}
M.~Kravtsev.
\newblock Quantum finite one-counter automata.
\newblock In \emph{SOFSEM'99: Theory and Practice of Computer Science}, volume
  1725 of \emph{Lecture Notes in Computer Science}, pages 431--440, 1999.

\bibitem[Li and Qiu(2008)]{LQ08}
L.~Li and D.~Qiu.
\newblock Determining the equivalence for one-way quantum finite automata.
\newblock \emph{Theoretical Computer Science}, 403\penalty0 (1):\penalty0
  42--51, 2008.
\newblock ISSN 0304-3975.

\bibitem[Moore and Crutchfield(2000)]{MC00}
C.~Moore and J.~P. Crutchfield.
\newblock Quantum automata and quantum grammars.
\newblock \emph{Theoretical Computer Science}, 237\penalty0 (1-2):\penalty0
  275--306, 2000.
\newblock ISSN 0304-3975.

\bibitem[Nasu and Honda(1971)]{NH71}
M.~Nasu and N.~Honda.
\newblock A context-free language which is not acceptable by a probabilistic
  automaton.
\newblock \emph{Information and Control}, 18\penalty0 (3):\penalty0 233--236,
  1971.

\bibitem[Nielsen and Chuang(2000)]{NC00}
M.~A. Nielsen and I.~L. Chuang.
\newblock \emph{Quantum Computation and Quantum Information}.
\newblock Cambridge University Press, 2000.

\bibitem[Paschen(2000)]{Pa00}
K.~Paschen.
\newblock Quantum finite automata using ancilla qubits.
\newblock Technical report, University of Karlsruhe, 2000.
\newblock Available at http://digbib.ubka.uni-karlsruhe.de/volltexte/1452000.

\bibitem[Pin(1987)]{Pi87}
J.-{\`E}. Pin.
\newblock On the language accepted by finite reversible automata.
\newblock In \emph{ICALP'87: Proceedings of the 14th International Colloquium,
  on Automata, Languages and Programming}, pages 237--249. Springer-Verlag,
  1987.

\bibitem[Say and Yakary{\i}lmaz(2010)]{SY10A}
A.~C.~C. Say and A.~Yakary{\i}lmaz.
\newblock Quantum function computation using sublogarithmic space, 2010.
\newblock (Poster presentation at QIP2010 - available at arXiv:1009.3124).

\bibitem[Say et~al.(2010)Say, Yakary{\i}lmaz, and Y\"{u}zsever]{SYY10A}
A.~C.~C. Say, A.~Yakary{\i}lmaz, and {\c S}.~Y\"{u}zsever.
\newblock Quantum one-way one-counter automata.
\newblock In R.~Freivalds, editor, \emph{Randomized and quantum computation},
  pages 25--34, 2010.
\newblock Satellite workshop of MFCS and CSL 2010.

\bibitem[Szepietowski(1994)]{Sz94}
A.~Szepietowski.
\newblock \emph{Turing Machines with Sublogarithmic Space}.
\newblock Springer-Verlag, 1994.

\bibitem[Watrous(1998)]{Wa98}
J.~Watrous.
\newblock \emph{Space-bounded quantum computation}.
\newblock PhD thesis, University of Wisconsin - Madison, USA, 1998.

\bibitem[Watrous(2003)]{Wa03}
J.~Watrous.
\newblock On the complexity of simulating space-bounded quantum computations.
\newblock \emph{Computational Complexity}, 12\penalty0 (1-2):\penalty0 48--84,
  2003.
\newblock ISSN 1016-3328.

\bibitem[Watrous(2009)]{Wa09}
J.~Watrous.
\newblock Quantum computational complexity.
\newblock In R.~A. Meyers, editor, \emph{Encyclopedia of Complexity and Systems
  Science}, pages 7174--7201. Springer, 2009.

\bibitem[Yakary{\i}lmaz(2010)]{Ya10A}
A.~Yakary{\i}lmaz.
\newblock \emph{Classical and Quantum Computation with Small Space Bounds}.
\newblock PhD thesis, Bo\u{g}azi\c{c}i University, 2010.
\newblock (Submitted).

\bibitem[Yakary{\i}lmaz and Say(2009)]{YS09B}
A.~Yakary{\i}lmaz and A.~C.~C. Say.
\newblock Efficient probability amplification in two-way quantum finite
  automata.
\newblock \emph{Theoretical Computer Science}, 410\penalty0 (20):\penalty0
  1932--1941, 2009.

\bibitem[Yakary{\i}lmaz and Say(2010{\natexlab{a}})]{YS10A}
A.~Yakary{\i}lmaz and A.~C.~C. Say.
\newblock Languages recognized by nondeterministic quantum finite automata.
\newblock \emph{Quantum Information and Computation}, 10\penalty0
  (9\&10):\penalty0 747--770, 2010{\natexlab{a}}.

\bibitem[Yakary{\i}lmaz and Say(2010{\natexlab{b}})]{YS10C}
A.~Yakary{\i}lmaz and A.~C.~C. Say.
\newblock Unbounded-error quantum computation with small space bounds.
\newblock Technical Report arXiv:1007.3624, 2010{\natexlab{b}}.

\bibitem[Yamasaki et~al.(2002)Yamasaki, Kobayashi, Tokunaga, and Imai]{YKTI02}
T.~Yamasaki, H.~Kobayashi, Y.~Tokunaga, and H.~Imai.
\newblock One-way probabilistic reversible and quantum one-counter automata.
\newblock \emph{Theoretical Computer Science}, 289\penalty0 (2):\penalty0
  963--976, 2002.

\bibitem[Yamasaki et~al.(2005)Yamasaki, Kobayashi, and Imai]{YKI05}
T.~Yamasaki, H.~Kobayashi, and H.~Imai.
\newblock Quantum versus deterministic counter automata.
\newblock \emph{Theoretical Computer Science}, 334\penalty0 (1-3):\penalty0
  275--297, 2005.

\bibitem[Yao(1993)]{Ya93}
A.~C.-C. Yao.
\newblock Quantum circuit complexity.
\newblock In \emph{SFCS'93: Proceedings of the 1993 IEEE 34th Annual
  Foundations of Computer Science}, pages 352--361, 1993.

\end{thebibliography}

\end{document}